\documentclass[12pt]{article}
\usepackage{amsmath,amssymb,amscd,amsthm,color,graphics,verbatim}
\usepackage[margin=0.95in]{geometry}

\usepackage[enableskew]{youngtab}
\def\qed{\hfill $\vrule height 2.5mm  width 2.5mm depth 0mm $}

\newtheorem{theorem}{Theorem}
\newtheorem{proposition}[theorem]{Proposition}
\newtheorem{lemma}[theorem]{Lemma}

\newtheorem{conjecture}[theorem]{Conjecture}

\theoremstyle{definition}

\newtheorem{definition}[theorem]{Definition}
\newtheorem{remark}[theorem]{Remark}

\begin{document}
$\,$\vspace{0mm}

\begin{center}
{\sf\huge Some Remarks On Nepomechie--Wang}
\vspace{3mm}\\
{\sf\huge Eigenstates For Spin 1/2\, XXX Model}
\vspace{13mm}\\
{\textsf{\LARGE  ${}^{\mbox{\small a,b}}$Anatol N. Kirillov and
${}^{\mbox{\small c}}$Reiho Sakamoto}}
\vspace{15mm}\\
{\textsf {${}^{\mbox{\small{a}}}$Research Institute for Mathematical Sciences,}}
\vspace{-1mm}\\
{\textsf {Kyoto University, Sakyo-ku,}}
\vspace{-1mm}\\
{\textsf {Kyoto, 606-8502, Japan}}
\vspace{-1mm}\\
{\textsf {kirillov@kurims.kyoto-u.ac.jp}}
\vspace{4mm}\\
{\textsf {${}^{\mbox{\small{b}}}$The Kavli Institute for the Physics and}}
\vspace{-1mm}\\
{\textsf {Mathematics of the Universe (IPMU),}}
\vspace{-1mm}\\
{\textsf {The University of Tokyo,}}
\vspace{-1mm}\\
{\textsf {Kashiwa, Chiba, 277-8583, Japan}}
\vspace{4mm}\\
{\textsf {${}^{\mbox{\small{c}}}$Department of Physics,}}
\vspace{-1mm}\\
{\textsf {Tokyo University of Science, Kagurazaka,}}
\vspace{-1mm}\\
{\textsf {Shinjuku, Tokyo, 162-8601, Japan}}
\vspace{-1mm}\\
{\textsf {reiho@rs.tus.ac.jp}}
\vspace{-1mm}\\
{\textsf {reihosan@08.alumni.u-tokyo.ac.jp}}
\vspace{20mm}\\
{\textit{\large To Professor Boris Feigin on the occasion}}\\
{\textit{\large of his sixtieth anniversary}}
\vspace{15mm}\\
\end{center}

\begin{abstract}
\noindent
We compute the energy eigenvalues of Nepomechie--Wang's eigenstates
for the spin 1/2 isotropic Heisenberg chain.\medskip
\end{abstract}

\pagebreak

\section{Introduction}
The Bethe ansatz \cite{Bethe} allows us to construct eigenvectors for Hamiltonians
of a wide range of integrable systems.
In our paper we are basically interested in the spin $\frac{1}{2}$ isotropic
Heisenberg chain (also known as the XXX model) with the periodic boundary condition which is the subject
of the original Bethe's paper.
According to the algebraic Bethe ansatz \cite{FT} (see also the book \cite{KorepinBook}),
the essence of the construction can be thought in the following way.
We start from a family of mutually commuting operators
$\{B_N(\lambda)\}_{\lambda\in\mathbb{C}}$, $[B_N(\lambda),B_N(\mu)]=0$,
and the ground state vector $|0\rangle_N$ where $N$ is the length of the chain.
Suppose that a collection of mutually distinct complex numbers
$\lambda=(\lambda_1,\lambda_2,\ldots,\lambda_\ell)$ satisfies the following
system of algebraic equations, commonly known as the Bethe ansatz equations,
\begin{align}
\label{eq:Bethe_ansatz}
\left(
\frac{\lambda_k+\frac{i}{2}}{\lambda_k-\frac{i}{2}}
\right)^N
=\prod_{j=1 \atop j\neq k}^\ell
\frac{\lambda_k-\lambda_j+i}{\lambda_k-\lambda_j-i},
\qquad
(k=1,\cdots,\ell),
\end{align}
then the vector
\begin{align}
\Psi_N(\lambda_1,\ldots,\lambda_\ell)=\Psi_{\lambda,N}
=B_N(\lambda_1)\cdots B_N(\lambda_\ell)|0\rangle_N
\end{align}
is an eigenvector of the spin $\frac{1}{2}$ isotropic Heisenberg chain
if the vector is non-zero.

Recall that a solution $\lambda=(\lambda_1,\lambda_2,\ldots,\lambda_\ell)$
to the equation (\ref{eq:Bethe_ansatz}) is called regular
if the corresponding vector is non-zero $\Psi_{\lambda,N}\neq 0$.
It had been observed by H. Bethe that the number of regular solutions
to the system (\ref{eq:Bethe_ansatz}) is strictly smaller than the number
of eigenvectors of the spin $\frac{1}{2}$ Heisenberg chain Hamiltonian
even for $N=4$ and $\ell=2$ case.
The problem to construct ``missing" eigenstates has been investigated
by many authors, and partially solved by \cite[Eq.(26)]{EKS}.
The most natural way to characterize and construct ``missing" eigenstates
has been developed by Nepomechie--Wang \cite{NW}.
Recall that for $N=4$ and $\ell=2$, a ``missing solution" corresponds
to solutions $(\lambda_1,\lambda_2)=(\frac{i}{2},-\frac{i}{2})$
in which case we have $B_4(\frac{i}{2})B_4(-\frac{i}{2})|0\rangle_4=0$.
The similar phenomena appears for general singular solutions to
the Bethe ansatz equations of the form
\begin{align}\label{eq:sing}
\lambda=\left\{\frac{i}{2},-\frac{i}{2},
\lambda_3,\ldots,\lambda_\ell\right\}.
\end{align}
The problem treated and partially solved in \cite{NW}
is to find a selection rule which guarantee that we can achieve
$\Psi_{\lambda,N}\neq 0$ under certain regularization.
Singular solutions of the form (\ref{eq:sing}) such that one can make
$\Psi_{\lambda,N}\neq 0$ is called physical singular solutions.
For $N\leq 14$, Nepomechie--Wang's rule is confirmed by
an extensive numerical computation \cite{HNS1}.
Also, the paper \cite{KS} reveals that the set of solutions which
satisfy Nepomechie--Wang's rule has a deep mathematical
structure called the rigged configurations (see Section \ref{sec:rc}).
The main purpose of the present paper is to give an explicit formula for
the energy eigenvalues of the Bethe vectors constructed from
the physical singular solutions (Theorem \ref{th:main}).
We also provide an alternative proof of results of \cite{NW}
at Proposition \ref{prop:NW}.

\section{Bethe vectors and Bethe ansatz equations}
To start with let us recall that the Bethe ansatz method is
a device to produce eigenvectors of an integrable system in question.
In the present paper we apply the Bethe ansatz method to the spin $\frac{1}{2}$
isotropic Heisenberg model under the periodic boundary condition.
The space of states $\mathfrak{H}_N$ of our model is
\begin{align}
\mathfrak{H}_N &= \bigotimes_{j=1}^{N} V_j,\quad V_j \simeq {\mathbb {C}}^2.
\end{align}
Then the Hamiltonian $\mathcal{H}_N$ is
\begin{align}
\label{def:Heisenberg_1/2}
\mathcal{H}_N &=  \frac{J}{4} \sum_{k=1}^{N}( \sigma_{k}^{1} \sigma_{k+1}^{1}
+ \sigma_{k}^{2} \sigma_{k+1}^{2}+ \sigma_{k}^{3} \sigma_{k+1}^{3} -{\mathbb {I}}_N),\qquad
\sigma_{N+1}^{a}=\sigma_{1}^{a},
\end{align}
where $\sigma^{a}$ $(a=1,2,3)$ are the Pauli matrices
\begin{align}
\sigma^{1} = 
\left(\!
\begin{array}{cc}
0&1\\
1&0
\end{array}
\!\right),\qquad
\sigma^{2} = 
\left(\!
\begin{array}{cc}
0&-i\\
i&0
\end{array}
\!\right),\qquad
\sigma^{3} = 
\left(\!
\begin{array}{cc}
1&0\\
0&-1
\end{array}
\!\right),
\end{align}
and the operators $\sigma_k^{a}$ $(a=1,2,3)$ act on $\mathfrak{H}_N$ as
\begin{align}
\sigma_{k}^{a}= I \otimes \cdots \otimes
\underbrace{\sigma^{a}}_{k}
\otimes \cdots \otimes I,
\end{align}
that is, they act non trivially only on the space $V_k$.
Here $I$ is the $2\times 2$ identity matrix and $\mathbb{I}_N$
is the identity matrix on the space of states; $\mathbb{I}_N=I^{\otimes N}$.

Let us consider the $L$-operators
\begin{align}
\label{eq:L_as_sum_of_sigmas}
L_k(\lambda)=\lambda I\otimes\mathbb{I}_N+
\frac{i}{2}\sum_a^3\sigma^a\otimes\sigma_k^a
\end{align}
which acts on $\mathbb{C}^2\otimes\mathfrak{H}_N$.
Then we define the transfer matrix
\begin{align}
T_N(\lambda)=L_N(\lambda)L_{N-1}(\lambda)\cdots L_1(\lambda).
\end{align}
The basic property of the $L$-operator (\ref{eq:L_as_sum_of_sigmas})
is that it satisfies the quantum Yang--Baxter relations, i.e.,
\begin{align}\label{eq:RLL=LLR}
R(\lambda-\mu)
\left(
L_{k}(\lambda)\otimes L_{k}(\mu)
\right)=
\left(
L_{k}(\mu)\otimes L_k(\lambda)
\right)
R(\lambda-\mu)
\end{align}
where
\begin{align}
R(\lambda)
=\frac{1}{\lambda +i}\left(
\left(\frac{\lambda}{2}+i\right)I\otimes I
+\frac{\lambda}{2}\sum_{a=1}^3\sigma^a\otimes\sigma^a
\right).
\end{align}

As a corollary of the quantum Yang--Baxter equation (\ref{eq:RLL=LLR}),
one can show that the transfer matrices $T_N(\lambda)$ and $T_N(\mu)$
commute for any parameters $\lambda$ and $\mu$.
It is clear from the definition of $L$-operator $L_k(\lambda)$,
see (\ref{eq:L_as_sum_of_sigmas}), that the transfer matrix
can be treated as $2\times 2$ matrix
\begin{align}
T_N(\lambda)=
\left(
\begin{array}{cc}
A_N(\lambda) & B_N(\lambda)\\
C_N(\lambda) & D_N(\lambda)
\end{array}
\right),
\end{align}
where $A_N(\lambda)$, $B_N(\lambda)$, $C_N(\lambda)$ and $D_N(\lambda)$
are operators acting on the space of states $\mathfrak{H}_N$.
The fundamental consequence of the fact that the transfer matrix $T_N(\lambda)$
also satisfies the quantum Yang--Baxter equation is that the operators
$B_N(\lambda)$ and $B_N(\mu)$ commute for any parameters $\lambda\in\mathbb{C}$
and $\mu\in\mathbb{C}$.

Let $\tau_N(\lambda)$ be the trace of the transfer matrix $T_N(\lambda)$
on the auxiliary space:
\begin{align}
\tau_N(\lambda)=A_N(\lambda)+D_N(\lambda).
\end{align}
Then the main observation of the algebraic Bethe ansatz analysis
of the XXX model is the following relation (see \cite{FT}):

\begin{theorem}\label{th:XXX_hamiltonian_tau}
We have
\begin{align}
\mathcal{H}_N=
\frac{iJ}{2}\frac{d}{d\lambda}\log \tau_N(\lambda)\Bigr|_{\lambda=\frac{i}{2}}
-\frac{NJ}{2}\mathbb{I}_N
\end{align}\qed
\end{theorem}

Now it is time to consider the local vectors $v_+=\left(\!
\begin{array}{c}
1\\0
\end{array}\!
\right)
\in V_k\simeq\mathbb{C}^2$ ($k=1,\cdots N$)
and the global one
\begin{align}
|0\rangle_N=v_+\otimes\cdots\otimes v_+\in\mathfrak{H}_N.
\end{align}
It is well-known that the vector $|0\rangle_N$ is an eigenvector of the operators
$A_N(\lambda)$, $D_N(\lambda)$ and $C_N(\lambda)$, namely,
\begin{align}
\label{eq:value_transfermatrix_A}
A_N(\lambda)|0\rangle_N&=\left(\lambda+\frac{i}{2}\right)^N|0\rangle_N,\\
\label{eq:value_transfermatrix_D}
D_N(\lambda)|0\rangle_N&=\left(\lambda-\frac{i}{2}\right)^N|0\rangle_N,\\
C_N(\lambda)|0\rangle_N&=0.
\end{align}

\begin{definition}
Define the Bethe vector corresponding to a collection of pairwise distinct
complex numbers $\lambda=(\lambda_1,\ldots,\lambda_\ell)$ as
\begin{align}
\Psi_N(\lambda_1,\ldots,\lambda_\ell)=B_N(\lambda_1)\cdots B_N(\lambda_\ell)|0\rangle_N.
\end{align}
\qed
\end{definition}

The basic property of the Bethe vectors is that $\Psi_N(\lambda_1,\ldots,\lambda_\ell)$
is an eigenvector of the operator $\tau_N(\lambda)$,
and thus of the Hamiltonian $\mathcal{H}_N$ (see Theorem \ref{th:XXX_hamiltonian_tau}) if and only if
\begin{enumerate}
\item
the parameters $\lambda_1,\ldots,\lambda_\ell$ satisfy the system of the Bethe
ansatz equations
\begin{align}
\left(
\frac{\lambda_k+\frac{i}{2}}{\lambda_k-\frac{i}{2}}
\right)^N
=\prod_{j=1 \atop j\neq k}^\ell
\frac{\lambda_k-\lambda_j+i}{\lambda_k-\lambda_j-i},
\qquad
(k=1,\cdots,\ell),
\end{align}
\item
and $\Psi_N(\lambda_1,\ldots,\lambda_\ell)\neq 0$.
\end{enumerate}
This result is derived from the action of $\tau_N(\lambda)$ on
the Bethe vectors $\Psi_N(\lambda_1,\ldots,\lambda_\ell)$.
More precisely, according to the standard argument (see, e.g., \cite{FT}),
we have the following expressions:
\begin{align}
\nonumber
&\{A_N(\lambda)+D_N(\lambda)\}B_N(\lambda_1)\cdots B_N(\lambda_\ell)|0\rangle_N\\
\label{eq:action_of_tN_on_Bethe_vectors}
&=
\Lambda(\lambda;\lambda_1,\cdots,\lambda_\ell)\prod_{j=1}^\ell B_N(\lambda_j)|0\rangle_N
+\sum_{k=1}^\ell\biggl\{
\Lambda_k(\lambda;\lambda_1,\cdots,\lambda_\ell)B_N(\lambda)
\prod_{j=1\atop j\neq k}^\ell B_N(\lambda_j)|0\rangle_N\biggr\},
\end{align}
where
\begin{align}
\label{eq:eigenvalue_of_tau_on_Bethe_vectors}
\Lambda(\lambda;\lambda_1,\cdots,\lambda_\ell)=
\left(
\lambda+\frac{i}{2}
\right)^N
\prod_{j=1}^\ell \frac{\lambda-\lambda_j-i}{\lambda-\lambda_j}+
\left(
\lambda-\frac{i}{2}
\right)^N
\prod_{j=1}^\ell \frac{\lambda_j-\lambda-i}{\lambda_j-\lambda}
\end{align}
and for $k=1,2,\ldots,k$ we have
\begin{align}
\label{eq:Lambda_k_for_Bethe_vectors}
\Lambda_k(\lambda;\lambda_1,\cdots,\lambda_\ell)=
\frac{i}{\lambda-\lambda_k}
\Biggl\{
\left(
\lambda_k+\frac{i}{2}
\right)^N
\prod_{j=1\atop j\neq k}^\ell \frac{\lambda_k-\lambda_j-i}{\lambda_k-\lambda_j}-
\left(
\lambda_k-\frac{i}{2}
\right)^N
\prod_{j=1\atop j\neq k}^\ell \frac{\lambda_j-\lambda_k-i}{\lambda_j-\lambda_k}
\Biggr\}.
\end{align}

We remark that combining the identity (\ref{eq:action_of_tN_on_Bethe_vectors})
and Theorem \ref{th:XXX_hamiltonian_tau},
we deduce that the energy eigenvalue $\mathcal{E}$ of the Hamiltonian $\mathcal{H}_N$
corresponding to the eigenvector $\Psi_N(\lambda_1,\ldots,\lambda_\ell)$ is
\begin{align}\label{eq:E}
\mathcal{E}=-\frac{J}{2}\sum_{j=1}^\ell\frac{1}{\lambda_j^2+\frac{1}{4}}
\end{align}
if $\lambda_j\neq\pm\frac{i}{2}$ for all $j=1,2,\ldots,\ell$.
It is well known that the Hamiltonian $\mathcal{H}_N$ commutes with the action
of the algebra $\mathfrak{sl}_2$ which acts on $\mathfrak{H}_N$.
In particular, the energy eigenvalue is constant for all eigenvectors
belonging to the same irreducible $\mathfrak{sl}_2$-module.
To be more precise, let $\mathbf{m}$ be the $m$-dimensional irreducible $\mathfrak{sl}_2$-module.
Suppose that we have a non-zero Bethe vector $\Psi_N(\lambda_1,\ldots,\lambda_\ell)$
constructed from the solutions $\lambda_1,\ldots,\lambda_\ell$ to the Bethe ansatz equations.
Then it is known that the vector $\Psi_N(\lambda_1,\ldots,\lambda_\ell)$ is the
highest weight vector of the module $\mathbf{m}$ where $m=N-2\ell+1$.

Now it is time to recall the definition of the Nepomechie--Wang eigenstates.
To begin with, recall that a solution to the Bethe ansatz equation is called singular,
if it has the form
\begin{align}\label{eq:sing2}
\lambda=\left\{\frac{i}{2},-\frac{i}{2},
\lambda_3,\ldots,\lambda_\ell\right\}.
\end{align}
Note that since $B_N(\frac{i}{2})B_N(-\frac{i}{2})=0$ in this case, we have
\[
\Psi_\lambda=B_N\!\left(\frac{i}{2}\right)\!B_N\!\left(-\frac{i}{2}\right)\!
B_N(\lambda_3)\cdots B_N(\lambda_\ell)|0\rangle_N=0,
\]
and the energy eigenvalue $\mathcal{E}$ of the state $\Psi_\lambda$ is divergent.
To resolve this problem, i.e., to construct a non-zero eigenvector of the Hamiltonian
(\ref{def:Heisenberg_1/2}), following \cite{NW} we define the perturbed  version of
(\ref{eq:sing2}) as follows:
\begin{align}
\label{eq:regularization}
\lambda_1=\frac{i}{2}+\epsilon+c\,\epsilon^N,\qquad
\lambda_2=-\frac{i}{2}+\epsilon.
\end{align}
We note that a similar regularization method is described in \cite{AV} and
\cite[Eq.(3.4)]{BMSZ}.

Let
\begin{align}
\Psi^{(\epsilon)}_\lambda:=
\frac{1}{\epsilon^N}
B_N\!\left(\frac{i}{2}+\epsilon+c\,\epsilon^N\right)\!B_N\!\left(-\frac{i}{2}+\epsilon\right)\!
B_N(\lambda_3)\cdots B_N(\lambda_\ell)|0\rangle_N.
\end{align}
Then we need to prove the following statement.
\begin{proposition}\label{prop:NW}
Suppose that $c$ is given by (\ref{eq:c1_NepomechieWang})
and (\ref{eq:c2_NepomechieWang}).

(1) The vector $\lim_{\epsilon\rightarrow 0}\Psi^{(\epsilon)}_\lambda=\Psi_\lambda$ is well-defined.

(2) $\Psi_\lambda$ is an eigenvector of $\mathcal{H}_N$.
\qed
\end{proposition}

\begin{remark}
From the compatibility condition of $c$ in (\ref{eq:c1_NepomechieWang})
and (\ref{eq:c2_NepomechieWang}), \cite{NW} deduce a criterion for
the singular solutions to provide non-zero Bethe vectors.
Their criterion is verified up to the case of $N\leq 14$ by an
extensive numerical computation \cite{HNS1}.
\qed
\end{remark}

Although these assertions are essentially proved in \cite{NW},
their normalization of $B_N(\lambda)$ is different from
the standard normalization used in the present paper.
Since this difference of the normalizations changes the structure of the proof,
we include some of the details of an alternative proof here.

Our proof of (1) is similar to the proof of
$B_N(\frac{i}{2}+\epsilon)B_N(-\frac{i}{2}+\epsilon)\sim\epsilon^N$ given in
Appendix A of \cite{NW}.\footnote{In \cite{NW}, our $B_N(\lambda)$ is denoted
by $\tilde{B}_N(\lambda)$.}
On the other hand, the proof of the statement corresponding to (1)
given in \cite{NW} is simpler.\footnote{However their proof seems
slightly incomplete since we have
$\tilde{B}_N(\lambda_1)\tilde{B}_N(\lambda_2)|0\rangle_N\neq
\tilde{B}_N(\lambda_1)|0\rangle_N\times\tilde{B}_N(\lambda_2)|0\rangle_N$.
Here $\tilde{B}_N(\lambda)$ stands for the $B_N$ operator in the normalization of \cite{NW}.}

For the proof of the statement (2), we prepare the following lemma.
Note that the following behaviors are different from the corresponding ones of
\cite{NW} since we are using a different normalization.
\begin{lemma}\label{lem:asymptotics}
We use the regularization of equation (\ref{eq:regularization}).

(a) If we take
\begin{align}\label{eq:c1_NepomechieWang}
c=-\frac{2}{i^{N+1}}\prod^\ell_{j=3}\frac{\lambda_j-\frac{3i}{2}}{\lambda_j+\frac{i}{2}},
\end{align}
then we have
\begin{align}
\Lambda_1(\lambda;\lambda_1,\cdots,\lambda_\ell)\sim
\frac{\epsilon^{N+1}}{\lambda-\frac{i}{2}-\epsilon-c\,\epsilon^N}.
\end{align}

(b) If we take
\begin{align}\label{eq:c2_NepomechieWang}
c=2i^{N+1}\prod^\ell_{j=3}\frac{\lambda_j+\frac{3i}{2}}{\lambda_j-\frac{i}{2}},
\end{align}
then we have
\begin{align}
\Lambda_2(\lambda;\lambda_1,\cdots,\lambda_\ell)\sim\frac{\epsilon^{N+1}}{\lambda+\frac{i}{2}-\epsilon}.
\end{align}
\end{lemma}
\begin{proof}
(a) We have
\begin{align*}
\Lambda_1&=
\frac{i}{\lambda-\lambda_1}
\Biggl\{
\left(
\lambda_1+\frac{i}{2}
\right)^N
\prod_{j=2}^\ell \frac{\lambda_1-\lambda_j-i}{\lambda_1-\lambda_j}-
\left(
\lambda_1-\frac{i}{2}
\right)^N
\prod_{j=2}^\ell \frac{\lambda_j-\lambda_1-i}{\lambda_j-\lambda_1}
\Biggr\}\\
&=
\frac{i}{\lambda-\lambda_1}
\Biggl\{
i^N\cdot\frac{c\,\epsilon^N}{i}
\prod_{j=3}^\ell \frac{\frac{i}{2}-\lambda_j-i}{\frac{i}{2}-\lambda_j}-
\epsilon^N\cdot\frac{-2i}{-i}
\prod_{j=3}^\ell \frac{\lambda_j-\frac{i}{2}-i}{\lambda_j-\frac{i}{2}}
\Biggr\}\\
&=
\frac{i\,\epsilon^N}{\lambda-\lambda_1}
\Biggl\{
c\cdot i^{N-1}
\prod_{j=3}^\ell \frac{\lambda_j+\frac{i}{2}}{\lambda_j-\frac{i}{2}}-
2
\prod_{j=3}^\ell \frac{\lambda_j-\frac{3i}{2}}{\lambda_j-\frac{i}{2}}
\Biggr\}.
\end{align*}
Therefore if we take $c$ as in (\ref{eq:c1_NepomechieWang}),
we see that $\Lambda_1\sim\epsilon^{N+1}/(\lambda-\lambda_1)$.

(b) We have
\begin{align*}
\Lambda_2&=
\frac{i}{\lambda-\lambda_2}
\Biggl\{
\left(
\lambda_2+\frac{i}{2}
\right)^N
\prod_{j=1\atop j\neq 2}^\ell \frac{\lambda_2-\lambda_j-i}{\lambda_2-\lambda_j}-
\left(
\lambda_2-\frac{i}{2}
\right)^N
\prod_{j=1\atop j\neq 2}^\ell \frac{\lambda_j-\lambda_2-i}{\lambda_j-\lambda_2}
\Biggr\}\\
&=
\frac{i}{\lambda-\lambda_2}
\Biggl\{
\epsilon^N\cdot\frac{-2i}{-i}
\prod_{j=3}^\ell \frac{-\frac{i}{2}-\lambda_j-i}{-\frac{i}{2}-\lambda_j}-
(-i)^N\,\frac{c\,\epsilon^N}{i}
\prod_{j=3}^\ell \frac{\lambda_j+\frac{i}{2}-i}{\lambda_j+\frac{i}{2}}
\Biggr\}\\
&=
\frac{i\,\epsilon^N}{\lambda-\lambda_2}
\Biggl\{
2\prod_{j=3}^\ell \frac{\lambda_j+\frac{3i}{2}}{\lambda_j+\frac{i}{2}}-
\frac{c}{i^{N+1}}
\prod_{j=3}^\ell \frac{\lambda_j-\frac{i}{2}}{\lambda_j+\frac{i}{2}}
\Biggr\}.
\end{align*}
Therefore if we take $c$ as in (\ref{eq:c2_NepomechieWang}),
we see that $\Lambda_1\sim\epsilon^{N+1}/(\lambda-\lambda_2)$.
\end{proof}
Applying the statements (a) and (b) of Lemma \ref{lem:asymptotics}
to identity (\ref{eq:action_of_tN_on_Bethe_vectors}),
we come to a proof of Proposition \ref{prop:NW} (2).

Finally let us give a remark on meanings of the regularization (\ref{eq:regularization}).
As we see in the present section, this regularization correctly provides
eigenvectors of the Hamiltonian $\mathcal{H}_N$.
Moreover we will show in the next two sections that the regularization (\ref{eq:regularization})
indeed provides the correct energy eigenvalues.
Recall that the Schr\"{o}dinger equation is the eigenvalue problem
$\mathcal{H}_N\Psi_\lambda =\mathcal{E}_\lambda\Psi_\lambda$.
Therefore we can justify the regularization (\ref{eq:regularization})
as it provides enough information pertaining to the singular states to the
Schr\"{o}dinger equation for the spin 1/2 isotropic Heisenberg model.

\section{Energy eigenvalues for the Nepomechie--Wang states}
Now we derive the energy eigenvalues for the Nepomechie--Wang states.
The main result is Theorem \ref{th:main}.

\paragraph{1)}
Let $\mathcal{E}$ be the energy eigenvalue corresponding to the
solutions $\{\lambda_1,\lambda_2,\ldots,\lambda_\ell\}$ of the Bethe ansatz equations.
From Theorem \ref{th:XXX_hamiltonian_tau},
we see that it is enough to compute
\begin{align}
\mathcal{E}=\frac{J}{2}
\Biggl\{
i\frac{d}{d\lambda}\log\Lambda(\lambda;\lambda_1,\cdots,\lambda_\ell)\biggr|_{\lambda=\frac{i}{2}}
-N\Biggr\}
\end{align}
where
\begin{align}
\Lambda(\lambda;\lambda_1,\cdots,\lambda_\ell)=
\left(
\lambda+\frac{i}{2}
\right)^N
\prod_{j=1}^\ell \frac{\lambda-\lambda_j-i}{\lambda-\lambda_j}+
\left(
\lambda-\frac{i}{2}
\right)^N
\prod_{j=1}^\ell \frac{\lambda_j-\lambda-i}{\lambda_j-\lambda}
\end{align}
as in (\ref{eq:eigenvalue_of_tau_on_Bethe_vectors}).
Thus it is enough to compute
\begin{align}
\varepsilon=i\frac{d}{d\lambda}\log\Lambda\biggr|_{\lambda=\frac{i}{2}}
=\frac{i\frac{d\Lambda}{d\lambda}}{\Lambda}\biggr|_{\lambda=\frac{i}{2}}.
\end{align}

\paragraph{2)}
Let us compute the denominator of $\varepsilon$:
\begin{align}
\varepsilon_{\rm deno}:=
\Lambda\left(\frac{i}{2};\lambda_1,\cdots,\lambda_\ell\right)=
i^N\prod_{j=1}^\ell\frac{\lambda_j+\frac{i}{2}}{\lambda_j-\frac{i}{2}}.
\end{align}
By using the regularizations (\ref{eq:regularization}), we obtain
\begin{align}
\varepsilon_{\rm deno}:=
i^N\cdot
\frac{i+\epsilon+c\,\epsilon^N}{\epsilon+c\,\epsilon^N}\cdot
\frac{\epsilon}{\epsilon-i}\,
\prod_{j=3}^\ell\frac{\lambda_j+\frac{i}{2}}{\lambda_j-\frac{i}{2}}
=i^N\cdot
\frac{i+\epsilon+c\,\epsilon^N}{(1+c\,\epsilon^{N-1})(\epsilon-i)}\,
\prod_{j=3}^\ell\frac{\lambda_j+\frac{i}{2}}{\lambda_j-\frac{i}{2}}.
\end{align}

\paragraph{3)}
By using the identity
\begin{align*}
\frac{d}{d\lambda}\frac{\lambda-\lambda_j-i}{\lambda-\lambda_j}
=\frac{d}{d\lambda}\!\left(1-\frac{i}{\lambda-\lambda_j}\right)
=\frac{i}{(\lambda_j-\lambda)^2},
\end{align*}
we have
\begin{align}
\nonumber
i\frac{d\Lambda}{d\lambda}=&\,
A_0(\lambda)+
\sum^\ell_{j=1}A_j(\lambda)\\
&+\text{terms containing at least one}\left(\lambda-\frac{i}{2}\right)
\end{align}
where
\begin{align*}
A_0(\lambda)=iN\!\left(\lambda+\frac{i}{2}\right)^{N-1}
\prod_{j=1}^\ell\frac{\lambda-\lambda_j-i}{\lambda-\lambda_j}
\end{align*}
and for $j=1,2,\ldots,\ell$,
\begin{align*}
A_j(\lambda)=
i\!\left(\lambda+\frac{i}{2}\right)^{N}
\frac{\lambda-\lambda_1-i}{\lambda-\lambda_1}\cdots
\frac{\lambda-\lambda_{j-1}-i}{\lambda-\lambda_{j-1}}\cdot
\frac{i}{(\lambda_j-\lambda)^2}\cdot
\frac{\lambda-\lambda_{j+1}-i}{\lambda-\lambda_{j+1}}\cdots
\frac{\lambda-\lambda_\ell-i}{\lambda-\lambda_\ell}.
\end{align*}
Below we compute the contribution from each term one by one.

\paragraph{4)}
Let us consider $A_0(\lambda)$:
\begin{align*}
A_0\!\left(\frac{i}{2}\right)&=
i^NN\cdot
\frac{\frac{i}{2}-(\frac{i}{2}+\epsilon+c\,\epsilon^N)-i}{\frac{i}{2}-(\frac{i}{2}+\epsilon+c\,\epsilon^N)}\cdot
\frac{\frac{i}{2}-(-\frac{i}{2}+\epsilon)-i}{\frac{i}{2}-(-\frac{i}{2}+\epsilon)}\,
\prod_{j=3}^\ell\frac{\lambda_j+\frac{i}{2}}{\lambda_j-\frac{i}{2}}\\
&=i^NN\cdot
\frac{i+\epsilon+c\,\epsilon^N}{(1+c\,\epsilon^{N-1})(\epsilon-i)}\,
\prod_{j=3}^\ell\frac{\lambda_j+\frac{i}{2}}{\lambda_j-\frac{i}{2}}.
\end{align*}
Therefore we obtain
\begin{align*}
\frac{1}{\varepsilon_{\rm deno}}\cdot
A_0\!\left(\frac{i}{2}\right)=N.
\end{align*}

\paragraph{5)}
Let us consider $A_1(\lambda)$ and $A_2(\lambda)$.
\begin{align*}
A_1\!\left(\frac{i}{2}\right)&=
i^{N+1}
\frac{i}{\{\frac{i}{2}-(\frac{i}{2}+\epsilon+c\,\epsilon^N)\}^2}\cdot
\frac{\frac{i}{2}-(-\frac{i}{2}+\epsilon)-i}{\frac{i}{2}-(-\frac{i}{2}+\epsilon)}\,
\prod_{j=3}^\ell\frac{\lambda_j+\frac{i}{2}}{\lambda_j-\frac{i}{2}}\\
&=-i^{N}
\frac{1}{\epsilon\,(1+c\,\epsilon^{N-1})^2(\epsilon-i)}\,
\prod_{j=3}^\ell\frac{\lambda_j+\frac{i}{2}}{\lambda_j-\frac{i}{2}}.
\end{align*}
On the other hand, we have
\begin{align*}
A_2\!\left(\frac{i}{2}\right)&=
i^{N+1}\,
\frac{\frac{i}{2}-(\frac{i}{2}+\epsilon+c\,\epsilon^N)-i}{\frac{i}{2}-(\frac{i}{2}+\epsilon+c\,\epsilon^N)}\cdot
\frac{i}{\{\frac{i}{2}-(-\frac{i}{2}+\epsilon)\}^2}\,
\prod_{j=3}^\ell\frac{\lambda_j+\frac{i}{2}}{\lambda_j-\frac{i}{2}}\\
&=-i^{N}
\frac{i+\epsilon+c\,\epsilon^N}{\epsilon\,(1+c\,\epsilon^{N-1})(\epsilon-i)^2}\,
\prod_{j=3}^\ell\frac{\lambda_j+\frac{i}{2}}{\lambda_j-\frac{i}{2}}.
\end{align*}
Thus we have
\begin{align*}
&\lim_{\epsilon\rightarrow 0}\,
\frac{1}{\varepsilon_{\rm deno}}
\left\{
A_1\!\left(\frac{i}{2}\right)+A_2\!\left(\frac{i}{2}\right)
\right\}\\
=&\lim_{\epsilon\rightarrow 0}\,
\frac{1}{\varepsilon_{\rm deno}}\times
\left(-i^{N}\right)
\frac{(\epsilon-i)+(i+\epsilon+c\,\epsilon^N)(1+c\,\epsilon^{N-1})}
{\epsilon\,(1+c\,\epsilon^{N-1})^2(\epsilon-i)^2}\,
\prod_{j=3}^\ell\frac{\lambda_j+\frac{i}{2}}{\lambda_j-\frac{i}{2}}\\
=&-\lim_{\epsilon\rightarrow 0}\,
\frac{(1+c\,\epsilon^{N-1})(\epsilon-i)}{i+\epsilon+c\,\epsilon^N}\times
\frac{2\epsilon+i\,c\,\epsilon^{N-1}+2c\,\epsilon^N+c^2\epsilon^{2N-1}}
{\epsilon\,(1+c\,\epsilon^{N-1})^2(\epsilon-i)^2}
\\
=&-\lim_{\epsilon\rightarrow 0}\,
\frac{2\epsilon+i\,c\,\epsilon^{N-1}+2c\,\epsilon^N+c^2\epsilon^{2N-1}}
{\epsilon\,(1+c\,\epsilon^{N-1})(\epsilon-i)(i+\epsilon+c\,\epsilon^N)}
=-2.
\end{align*}

\paragraph{6)}
Finally, for $j=3,4,\ldots,\ell$, we have
\begin{align*}
A_j\!\left(\frac{i}{2}\right)=&\,
i^{N+1}\,\frac{i+\epsilon+c\,\epsilon^N}{\epsilon+c\,\epsilon^N}\cdot
\frac{\epsilon}{\epsilon-i}\cdot
\frac{\lambda_3+\frac{i}{2}}{\lambda_3-\frac{i}{2}}\cdots
\frac{\lambda_{j-1}+\frac{i}{2}}{\lambda_{j-1}-\frac{i}{2}}\cdot
\frac{i}{(\lambda_j-\frac{i}{2})^2}\cdot
\frac{\lambda_{j+1}+\frac{i}{2}}{\lambda_{j+1}-\frac{i}{2}}\cdots
\frac{\lambda_\ell+\frac{i}{2}}{\lambda_\ell-\frac{i}{2}}.
\end{align*}
Thus we have
\begin{align*}
\frac{1}{\varepsilon_{\rm deno}}\cdot
A_j\!\left(\frac{i}{2}\right)=
-\,\frac{\lambda_j-\frac{i}{2}}{\lambda_j+\frac{i}{2}}\cdot
\frac{1}{(\lambda_j-\frac{i}{2})^2}=
-\,\frac{1}{\lambda_j^2+\frac{1}{4}}.
\end{align*}

\paragraph{7)}
To summarize, we have
\begin{align*}
\varepsilon=N-2-\sum_{j=3}^\ell\frac{1}{\lambda_j^2+\frac{1}{4}}.
\end{align*}
Therefore we obtain the following result.

\begin{theorem}\label{th:main}
Suppose that we have the following physical singular solutions to
the Bethe ansatz equations
\[
\lambda=
\left\{
\frac{i}{2},-\frac{i}{2},\lambda_3,\cdots,\lambda_\ell
\right\}.
\]
If we impose the regularization (\ref{eq:regularization}),
the corresponding non-zero Bethe vector (i.e., the Nepomechie--Wang state)
has the following energy eigenvalue:
\[
\mathcal{E}_\lambda=-J-\frac{J}{2}\sum_{j=3}^\ell\frac{1}{\lambda_j^2+\frac{1}{4}}.
\]\qed
\end{theorem}

\begin{remark}
In \cite{NW} they considered another regularization
\begin{align}
\lambda^\text{naive}_1=\frac{i}{2}+\epsilon,\qquad
\lambda^\text{naive}_1=-\frac{i}{2}+\epsilon
\end{align}
which they call the naive regularization.
However, our point is that we have to use the sophisticated regularization (\ref{eq:regularization})
for the computation of the energy eigenvalues.
Recall the form of the Schr\"{o}dinger equation
$\mathcal{H}_N\Psi_\lambda=\mathcal{E}_\lambda\Psi_\lambda$.
As we can see in the Introduction to \cite{NW}, the naive regularization
provides incorrect answers to the eigenvectors $\Psi_\lambda$.
In this and the next sections we demonstrate that the regularization (\ref{eq:regularization})
provides the correct energy eigenvalue $\mathcal{E}_\lambda$
corresponding to the correct eigenvector $\Psi_\lambda$,
thus completing the solutions to the Schr\"{o}dinger equation in the present case.
\qed
\end{remark}

\section{Examples}

\subsection{Rigged configurations}\label{sec:rc}
In our previous paper \cite{KS}, we pointed out that
the rigged configurations (RC for short) provide a good
parametrization for the combination of both regular solutions
and physical singular solutions to the Bethe ansatz equations.
In particular, we pointed out that the rigged configurations are essential
for the description of the physical singular solutions and,
as the result, we proposed conjectures on the total numbers of
various classes of solutions to the Bethe ansatz equations.

In the spin $1/2$ XXX model case, a rigged configuration is
comprised of a Young diagram $\nu$ (called a configuration) and
integers (called riggings) attached to each row of $\nu$.
To be more specific, let $\nu=(\nu_1,\nu_2,\ldots,\nu_g)$
and let $J_i$ ($1\leq i\leq g$) be the integer
attached to the length $\nu_i$ row of $\nu$.
Then the set of rigged configurations is comprised of all $\nu$
and $J_i$ ($1\leq i\leq g$) satisfying the following conditions.
Suppose that the length of the state is $N$.
Then the total number of the boxes of $\nu$ must not
exceed $N/2$.
We introduce the following integers which we call the vacancy numbers:
\begin{align}
P_k(\nu)=N-2\sum^g_{i=1}\min(k,\nu_i)\qquad (k\in\mathbb{Z}_{>0}).
\end{align}
Note that the second term is the number of boxes within the left
$k$ columns of $\nu$.
Suppose that the rigging $J_i$ is attached to a length $k$ row.
Then it must satisfy $0\leq J_i\leq P_k(\nu)$.
Note that, as rigged configurations, we do not make distinction
if the difference is only a reordering of riggings for the rows of the same length.

Below we provide labels of the solutions to the Bethe ansatz equations
in terms of the rigged configurations.
We refer the readers to \cite[Section 3.1]{KS} for the description
of the correspondence between the rigged configurations and the solutions
to the Bethe ansatz equations.

\subsection{$N=4$ case}
\paragraph{Regular solutions to the Bethe ansatz equations.}
We refer the readers to \cite[Example 2]{KS} for additional information on this case.
Since we consider regular solutions, we use (\ref{eq:E}) in order to determine
the energy eigenvalue $\mathcal{E}$.

\begin{itemize}
\item
The case $\ell=0$.
This case corresponds to the representation $\mathbf{5}$
which is generated by the vacuum vector $|0\rangle_4$.
Then we have $\mathcal{E}=0$.
\item
The case $\ell=1$.
The corresponding representation is $\mathbf{3}$.
Then we have the following three solutions:
\begin{center}
\begin{tabular}{lll}
\hline\hline
$\lambda_1$ & $\mathcal{E}$ & RC\\
\hline
$\frac{1}{2}$ & $-J$&
{\unitlength 12pt
\begin{picture}(3,1)
\put(0.2,0.1){2}
\put(1,0){\Yboxdim12pt\yng(1)}
\put(2.3,0.1){2}
\end{picture}
}\rule{0pt}{15pt}\\
0 & $-2J$&
{\unitlength 12pt
\begin{picture}(3,1)
\put(0.2,0.1){2}
\put(1,0){\Yboxdim12pt\yng(1)}
\put(2.3,0.1){1}
\end{picture}
}\rule{0pt}{15pt}\\
$-\frac{1}{2}$ & $-J$&
{\unitlength 12pt
\begin{picture}(3,1)
\put(0.2,0.1){2}
\put(1,0){\Yboxdim12pt\yng(1)}
\put(2.3,0.1){0}
\end{picture}
}\rule{0pt}{15pt}\\
\hline
\end{tabular}
\end{center}
Here, in order to display the rigged configurations,
we put the vacancy numbers (resp. riggings)
on the left (resp. right) of the corresponding rows of $\nu$.
\item
The case $\ell=2$.
The corresponding representation is $\mathbf{1}$.
Then we have only one regular solution:
\begin{center}
\begin{tabular}{lll}
\hline\hline
$\lambda_1, \lambda_2$ & $\mathcal{E}$ &RC\\
\hline
$\displaystyle\frac{1}{\sqrt{12}}, -\frac{1}{\sqrt{12}}$ & $-3J$&
{\unitlength 12pt
\begin{picture}(4,1)(0,0.6)
\put(0.2,1.1){0}
\put(0.2,0.1){0}
\put(1,0){\Yboxdim12pt\yng(1,1)}
\put(2.4,1.1){0}
\put(2.4,0.1){0}
\end{picture}
}\rule{0pt}{21pt}
\\\hline
\end{tabular}
\end{center}
\end{itemize}

To summarize, we have the following energy eigenvalues and their multiplicities
\[
\left\{0^5,(-J)^6,(-2J)^3,(-3J)^1\right\}
\]
from the regular solutions.
Here we describe the multiplicities of the energy eigenvalues
as in the following notation:
\[
\{\text{eigenvalue}^\text{multiplicity},\ldots\}.
\]

\paragraph{Direct diagonalization.}
From the exact diagonalization of $\mathcal{H}_4$,
we obtain the following multiplicities for the energy eigenvalues:
\[
\left\{0^5,(-J)^7,(-2J)^3,(-3J)^1\right\}.
\]
In conclusion, one eigenstate of eigenvalue $-J$ is missing from the list of regular solutions.

\paragraph{Nepomechie--Wang state.}
In the case of $\ell=2$ we have the following physical singular solution:
\begin{center}
\begin{tabular}{ll}
\hline\hline
$\lambda_1, \lambda_2$ & RC\\
\hline
$\displaystyle\frac{i}{2},-\frac{i}{2}$ &
{\unitlength 12pt
\begin{picture}(4,1)
\put(0.2,0.1){0}
\put(1,0){\Yboxdim12pt\yng(2)}
\put(3.3,0.1){0}
\end{picture}
}\rule{0pt}{18pt}
\\\hline
\end{tabular}
\end{center}
According to Theorem \ref{th:main}, the corresponding energy eigenvalue is $\mathcal{E}=-J$.
Moreover the corresponding representation is $\mathbf{1}$ since $\ell=2$.
This result is compatible with the above observations.

\subsection{$N=6$ case}
\paragraph{Regular solutions to the Bethe ansatz equations.}
We refer the readers to \cite[Example 11]{KS} for additional information on this case.
\begin{itemize}
\item
The case $\ell=0$.
This case corresponds to the representation $\mathbf{7}$
which is generated by the vacuum vector $|0\rangle_6$.
Then we have $\mathcal{E}=0$.
\item
The case $\ell=1$.
The corresponding representation is $\mathbf{5}$.
Then we have the following five solutions:
\begin{center}
\begin{tabular}{lll}
\hline\hline
$\lambda_1$ & $\mathcal{E}$ & RC\\
\hline
$0.866025$ & $-0.5J$&
{\unitlength 12pt
\begin{picture}(3,1)
\put(0.2,0.1){4}
\put(1,0){\Yboxdim12pt\yng(1)}
\put(2.3,0.1){4}
\end{picture}
}\rule{0pt}{15pt}\\
$0.288675$ & $-1.5J$&
{\unitlength 12pt
\begin{picture}(3,1)
\put(0.2,0.1){4}
\put(1,0){\Yboxdim12pt\yng(1)}
\put(2.3,0.1){3}
\end{picture}
}\rule{0pt}{15pt}\\
0 & $-2J$&
{\unitlength 12pt
\begin{picture}(3,1)
\put(0.2,0.1){4}
\put(1,0){\Yboxdim12pt\yng(1)}
\put(2.3,0.1){2}
\end{picture}
}\rule{0pt}{15pt}\\
$-0.288675$ & $-1.5J$&
{\unitlength 12pt
\begin{picture}(3,1)
\put(0.2,0.1){4}
\put(1,0){\Yboxdim12pt\yng(1)}
\put(2.3,0.1){1}
\end{picture}
}\rule{0pt}{15pt}\\
$-0.866025$ & $-0.5J$&
{\unitlength 12pt
\begin{picture}(3,1)
\put(0.2,0.1){4}
\put(1,0){\Yboxdim12pt\yng(1)}
\put(2.3,0.1){0}
\end{picture}
}\rule{0pt}{15pt}\\
\hline
\end{tabular}
\end{center}
\item
The case $\ell=2$.
The corresponding representation is $\mathbf{3}$.
Then we have the following eight regular solutions:\footnote{In order to
make the correspondence between the final six solutions
and rigged configurations in a clearly visible form, we plot solutions (labeled by $1,2,\ldots,6$)
on the complex plane.
\begin{center}
\unitlength 12pt
\begin{picture}(7,6)
\put(0.2,5.2){$1$}
\put(0,3){\vector(1,0){6}}
\put(3,0){\vector(0,1){6}}
\put(5.06,3){\circle*{0.3}}
\put(0.94,3){\circle*{0.3}}
\multiput(1,0.21)(0,0.2){29}{\circle*{0.07}}
\multiput(2,0.21)(0,0.2){29}{\circle*{0.07}}
\multiput(4,0.21)(0,0.2){29}{\circle*{0.07}}
\multiput(5,0.21)(0,0.2){29}{\circle*{0.07}}
\multiput(0.2,1.01)(0.2,0){29}{\circle*{0.07}}
\multiput(0.2,2.01)(0.2,0){29}{\circle*{0.07}}
\multiput(0.2,4.01)(0.2,0){29}{\circle*{0.07}}
\multiput(0.2,5.01)(0.2,0){29}{\circle*{0.07}}
\end{picture}
\begin{picture}(7,6)
\put(0.2,5.2){$2$}
\put(0,3){\vector(1,0){6}}
\put(3,0){\vector(0,1){6}}
\put(4.89,3){\circle*{0.3}}
\put(2.4,3){\circle*{0.3}}
\multiput(1,0.21)(0,0.2){29}{\circle*{0.07}}
\multiput(2,0.21)(0,0.2){29}{\circle*{0.07}}
\multiput(4,0.21)(0,0.2){29}{\circle*{0.07}}
\multiput(5,0.21)(0,0.2){29}{\circle*{0.07}}
\multiput(0.2,1.01)(0.2,0){29}{\circle*{0.07}}
\multiput(0.2,2.01)(0.2,0){29}{\circle*{0.07}}
\multiput(0.2,4.01)(0.2,0){29}{\circle*{0.07}}
\multiput(0.2,5.01)(0.2,0){29}{\circle*{0.07}}
\end{picture}
\begin{picture}(7,6)
\put(0.2,5.2){$3$}
\put(0,3){\vector(1,0){6}}
\put(3,0){\vector(0,1){6}}
\put(4.74,3){\circle*{0.3}}
\put(2.71,3){\circle*{0.3}}
\multiput(1,0.21)(0,0.2){29}{\circle*{0.07}}
\multiput(2,0.21)(0,0.2){29}{\circle*{0.07}}
\multiput(4,0.21)(0,0.2){29}{\circle*{0.07}}
\multiput(5,0.21)(0,0.2){29}{\circle*{0.07}}
\multiput(0.2,1.01)(0.2,0){29}{\circle*{0.07}}
\multiput(0.2,2.01)(0.2,0){29}{\circle*{0.07}}
\multiput(0.2,4.01)(0.2,0){29}{\circle*{0.07}}
\multiput(0.2,5.01)(0.2,0){29}{\circle*{0.07}}
\end{picture}
\end{center}

\begin{center}
\unitlength 12pt
\begin{picture}(7,6)
\put(0.2,5.2){$4$}
\put(0,3){\vector(1,0){6}}
\put(3,0){\vector(0,1){6}}
\put(3.59,3){\circle*{0.3}}
\put(1.10,3){\circle*{0.3}}
\multiput(1,0.21)(0,0.2){29}{\circle*{0.07}}
\multiput(2,0.21)(0,0.2){29}{\circle*{0.07}}
\multiput(4,0.21)(0,0.2){29}{\circle*{0.07}}
\multiput(5,0.21)(0,0.2){29}{\circle*{0.07}}
\multiput(0.2,1.01)(0.2,0){29}{\circle*{0.07}}
\multiput(0.2,2.01)(0.2,0){29}{\circle*{0.07}}
\multiput(0.2,4.01)(0.2,0){29}{\circle*{0.07}}
\multiput(0.2,5.01)(0.2,0){29}{\circle*{0.07}}
\end{picture}
\begin{picture}(7,6)
\put(0.2,5.2){$5$}
\put(0,3){\vector(1,0){6}}
\put(3,0){\vector(0,1){6}}
\put(3.48,3){\circle*{0.3}}
\put(2.51,3){\circle*{0.3}}
\multiput(1,0.21)(0,0.2){29}{\circle*{0.07}}
\multiput(2,0.21)(0,0.2){29}{\circle*{0.07}}
\multiput(4,0.21)(0,0.2){29}{\circle*{0.07}}
\multiput(5,0.21)(0,0.2){29}{\circle*{0.07}}
\multiput(0.2,1.01)(0.2,0){29}{\circle*{0.07}}
\multiput(0.2,2.01)(0.2,0){29}{\circle*{0.07}}
\multiput(0.2,4.01)(0.2,0){29}{\circle*{0.07}}
\multiput(0.2,5.01)(0.2,0){29}{\circle*{0.07}}
\end{picture}
\begin{picture}(7,6)
\put(0.2,5.2){$6$}
\put(0,3){\vector(1,0){6}}
\put(3,0){\vector(0,1){6}}
\put(3.28,3){\circle*{0.3}}
\put(1.25,3){\circle*{0.3}}
\multiput(1,0.21)(0,0.2){29}{\circle*{0.07}}
\multiput(2,0.21)(0,0.2){29}{\circle*{0.07}}
\multiput(4,0.21)(0,0.2){29}{\circle*{0.07}}
\multiput(5,0.21)(0,0.2){29}{\circle*{0.07}}
\multiput(0.2,1.01)(0.2,0){29}{\circle*{0.07}}
\multiput(0.2,2.01)(0.2,0){29}{\circle*{0.07}}
\multiput(0.2,4.01)(0.2,0){29}{\circle*{0.07}}
\multiput(0.2,5.01)(0.2,0){29}{\circle*{0.07}}
\end{picture}
\end{center}
Here the spacing of the dotted lines is $0.33$ and solutions
are arranged in the descending order of $\lambda_1$.
According to \cite{KS}, we assume that the upper riggings specify the positions of $\lambda_1$
(the larger rigging corresponds to the larger value of $\lambda_1$).
Next we specify $\lambda_2$ in the same manner.
In fact, this kind of a clear relation is a typical behavior.
See \cite[Section 4.1]{KS} for another example.
}
\begin{center}
\begin{tabular}{llll}
\hline\hline
label&$\lambda_1, \lambda_2$ & $\mathcal{E}$ &RC\\
\hline
&$0.554592\pm 0.512465i$ & $-0.7192J$&
{\unitlength 12pt
\begin{picture}(4,1)
\put(0.2,0.1){2}
\put(1,0){\Yboxdim12pt\yng(2)}
\put(3.3,0.1){2}
\end{picture}
}\rule{0pt}{15pt}
\\
&$-0.554592\pm 0.512465i$ & $-0.7192J$&
{\unitlength 12pt
\begin{picture}(4,1)
\put(0.2,0.1){2}
\put(1,0){\Yboxdim12pt\yng(2)}
\put(3.3,0.1){0}
\end{picture}
}\rule{0pt}{15pt}
\\
\hline
1&$0.688190,-0.688190$ & $-1.3819J$&
{\unitlength 12pt
\begin{picture}(4,1)
\put(0.2,1.1){2}
\put(0.2,0.1){2}
\put(1,0){\Yboxdim12pt\yng(1,1)}
\put(2.4,1.1){2}
\put(2.4,0.1){0}
\end{picture}
}\rule{0pt}{28pt}
\\
2&$0.631084,-0.198071$ & $-2.5J$&
{\unitlength 12pt
\begin{picture}(4,1)
\put(0.2,1.1){2}
\put(0.2,0.1){2}
\put(1,0){\Yboxdim12pt\yng(1,1)}
\put(2.4,1.1){2}
\put(2.4,0.1){1}
\end{picture}
}\rule{0pt}{28pt}
\\
3&$0.582004,-0.094167$ & $-2.7807J$&
{\unitlength 12pt
\begin{picture}(4,1)
\put(0.2,1.1){2}
\put(0.2,0.1){2}
\put(1,0){\Yboxdim12pt\yng(1,1)}
\put(2.4,1.1){2}
\put(2.4,0.1){2}
\end{picture}
}\rule{0pt}{28pt}
\\
4&$0.198071,-0.631084$ & $-2.5J$&
{\unitlength 12pt
\begin{picture}(4,1)
\put(0.2,1.1){2}
\put(0.2,0.1){2}
\put(1,0){\Yboxdim12pt\yng(1,1)}
\put(2.4,1.1){1}
\put(2.4,0.1){0}
\end{picture}
}\rule{0pt}{28pt}
\\
5&$0.162459,-0.162459$ & $-3.6180J$&
{\unitlength 12pt
\begin{picture}(4,1)
\put(0.2,1.1){2}
\put(0.2,0.1){2}
\put(1,0){\Yboxdim12pt\yng(1,1)}
\put(2.4,1.1){1}
\put(2.4,0.1){1}
\end{picture}
}\rule{0pt}{28pt}
\\
6&$0.094167,-0.582004$ & $-2.7807J$&
{\unitlength 12pt
\begin{picture}(4,1)
\put(0.2,1.1){2}
\put(0.2,0.1){2}
\put(1,0){\Yboxdim12pt\yng(1,1)}
\put(2.4,1.1){0}
\put(2.4,0.1){0}
\end{picture}
}\rule{0pt}{28pt}
\\
\hline
\end{tabular}
\end{center}
\item
The case $\ell=3$.
The corresponding representation is $\mathbf{1}$.
Then we have the following four regular solutions:
\begin{center}
\begin{tabular}{lll}
\hline\hline
$\lambda_1, \lambda_2, \lambda_3$ & $\mathcal{E}$ &RC\\
\hline
$0,\pm 1.008757i$& $-0.6972J$ &
{\unitlength 12pt
\begin{picture}(5,1.5)(-0.7,0)
\put(-0.7,0.1){$0$}
\put(0.1,0){\Yboxdim12pt\yng(3)}
\put(3.5,0.1){$0$}
\end{picture}}\\
$0.235900\pm 0.500280i, -0.471800$& $-2J$ &
{\unitlength 12pt
\begin{picture}(2.5,2.5)(-0.7,0)
\put(-0.7,1.1){$0$}
\put(-0.7,0.1){$2$}
\put(0.1,0){\Yboxdim12pt\yng(2,1)}
\put(2.4,1.1){$0$}
\put(1.3,0.1){$0$}
\end{picture}}\\%%
$0.471800,-0.235900\pm 0.500280i$& $-2J$ &
{\unitlength 12pt
\begin{picture}(2.5,2.5)(-0.7,0)
\put(-0.7,1.1){$0$}
\put(-0.7,0.1){$2$}
\put(0.1,0){\Yboxdim12pt\yng(2,1)}
\put(2.4,1.1){$0$}
\put(1.3,0.1){$2$}
\end{picture}}\\
$0,\pm 0.429253$& $-4.3027$ &
{\unitlength 12pt
\begin{picture}(1.5,3.5)(-0.7,0)
\put(-0.7,2.2){$0$}
\put(-0.7,1.1){$0$}
\put(-0.7,0.1){$0$}
\put(0.1,0){\Yboxdim12pt\yng(1,1,1)}
\put(1.4,2.2){$0$}
\put(1.4,1.1){$0$}
\put(1.4,0.1){$0$}
\end{picture}}\\\hline
\end{tabular}
\end{center}
\end{itemize}

To summarize, we have the following energy eigenvalues and their multiplicities
\begin{align*}
&\left\{
0^7, (-0.5J)^{10}, (-0.6972J)^1, (-0.7192J)^6, (-1.3819J)^3, (-1.5J)^{10},\right.\\
&\left. (-2J)^7, (-2.5J)^6, (-2.7807J)^6, (-3.6180J)^3, (-4.3027J)^1
\right\}
\end{align*}
from the regular solutions.

\paragraph{Direct diagonalization.}
From the exact diagonalization of $\mathcal{H}_6$,
we obtain the following multiplicities for the energy eigenvalues:
\begin{align*}
&\left\{0^7,
\left(-\frac{1}{2}J\right)^{10},
\left(-\frac{5-\sqrt{13}}{2} J\right)^1,
\left(-\frac{7-\sqrt{17}}{4}J\right)^6,
(-J)^3,
\left(-\frac{5-\sqrt{5}}{2}J\right)^3,
\left(-\frac{3}{2}J\right)^{10},
\right.\\
&\left.
(-2J)^7,
\left(-\frac{5}{2}J\right)^6,
\left(-\frac{7+\sqrt{17}}{4}J\right)^6,
(-3J)^1,
\left(-\frac{5+\sqrt{5}}{2}J\right)^3,
\left(-\frac{5+\sqrt{13}}{2}J\right)^1
\right\},
\end{align*}
or, in order to facilitate the comparison, their numerical values are
\begin{align*}
&\left\{
0^7, (-0.5J)^{10}, (-0.6972J)^1, (-0.7192J)^6, (-J)^3, (-1.3819J)^3, (-1.5J)^{10},\right.\\
&\left. (-2J)^7, (-2.5J)^6, (-2.7807J)^6, (-3J)^1, (-3.6180J)^3, (-4.3027J)^1
\right\}.
\end{align*}
In conclusion, the following energy eigenvalues (with multiplicities)
are missing from the list of the regular solutions:
\[
\left\{(-J)^3,(-3J)^1\right\}.
\]

\paragraph{Nepomechie--Wang state.}
In the case of $\ell=2$ we have the following physical singular solution
which generates the representation $\mathbf{3}$:
\begin{center}
\begin{tabular}{ll}
\hline\hline
$\lambda_1, \lambda_2$ & RC\\
\hline
$\displaystyle\frac{i}{2},-\frac{i}{2}$ &
{\unitlength 12pt
\begin{picture}(4,1)
\put(0.2,0.1){2}
\put(1,0){\Yboxdim12pt\yng(2)}
\put(3.3,0.1){1}
\end{picture}
}\rule{0pt}{18pt}
\\\hline
\end{tabular}
\end{center}
According to Theorem \ref{th:main}, the corresponding energy eigenvalue is $\mathcal{E}=-J$.

In the case of $\ell=3$, we have the following physical singular solution
which generates the representation $\mathbf{1}$:
\begin{center}
\begin{tabular}{ll}
\hline\hline
$\lambda_1, \lambda_2, \lambda_3$ & RC\\
\hline
$0, \displaystyle\frac{i}{2},-\frac{i}{2}$ &
{\unitlength 12pt
\begin{picture}(4,1)(0,0.6)
\put(0.2,1.1){0}
\put(0.2,0.1){2}
\put(1,0){\Yboxdim12pt\yng(2,1)}
\put(3.3,1.1){0}
\put(2.3,0.1){1}
\end{picture}
}\rule{0pt}{20pt}
\\\hline
\end{tabular}
\end{center}
According to Theorem \ref{th:main}, the corresponding energy eigenvalue is $\mathcal{E}=-3J$.

Thus we have a perfect agreement with the above observations.

\section{Conclusion}
\paragraph{(a)} We compute the energy of the Nepomechie--Wang eigenstates
which correspond to the physical singular solutions to the
Bethe ansatz equations (Theorem \ref{th:main}).
Recall that in our previous paper \cite{KS}, we pointed out that
the set of solutions to the Bethe ansatz equations which are either
regular or physical singular in the sense of \cite{NW} has a
deep mathematical structure called the rigged configurations.
Such property is apparent even for smaller values of the system size $N$.
Therefore the present result provides yet another supporting evidence
for the usefulness of Nepomechie--Wang's results.

We remark that in paper \cite{EKS}, the authors find examples
where some of the string type solutions are replaced by pairs of real solutions.
Therefore we expect that the correspondence between the rigged configurations
and the set of regular and physical singular solutions to the Bethe ansatz equations
requires extra modifications when $N$ is large.

\paragraph{(b)} We expect interesting connections between the physical singular solutions
to the spin $\frac{1}{2}$ isotropic Heisenberg model and anomaly dimensions
of certain generic gauge invariant operators in AdS$\times S^5$ theory
studied in \cite{BMSZ}.

\paragraph{(c)} In \cite{HNS2}, the authors considered the spin-$s$ generalized Heisenberg chain.
According to their numerical data, we propose the following conjecture.

\begin{conjecture}
(1) If $2s\equiv 1\pmod 2$, then the total number of states consists of either
regular solutions or physical singular solutions (i.e., there are no strange solutions,
i.e., solutions to the Bethe ansatz equations having some components equal,
and therefore violate the Pauli principle), except possibly ``sporadic physical states"
to the Bethe ansatz equations\footnote{Indeed, the Bethe ansatz equations for
spin $s$ Heisenberg chain have the following form:
\begin{align}
\left(\frac{\lambda_k+is}{\lambda_k-is}\right)^N=
\prod^\ell_{j=1\atop j\neq k}\frac{\lambda_k-\lambda_j+i}{\lambda_k-\lambda_j-i},\qquad
(k=1,\ldots,\ell).
\end{align}
Therefore, if $\lambda=\lambda^{(\ell)}_\pm$, then
$\left(\frac{\pm 1+i}{\pm 1-i}\right)^N=(-1)^{\ell-1}$, or equivalently,
$(\mp i)^N=(-1)^{\ell-1}$, so that, $N\equiv 2\ell-2\pmod 4$;
In the case $\lambda=\lambda^{(\ell)}_0$, the Bethe ansatz equations take the form
$(-1)^N=(-1)^{\ell-1}$, i.e., $N\equiv\ell-1\pmod 2$.}
\begin{align}
\lambda^{(\ell)}_0&=(\underbrace{0,\ldots,0}_\ell),\qquad\mbox{if } N\equiv\ell-1\pmod 2,\\
\lambda^{(\ell)}_\pm&=(\underbrace{\pm s,\ldots,\pm s}_\ell),\qquad\mbox{if } N\equiv 2\ell-2\pmod 4.
\end{align}
As it has been shown in \cite[Table 5]{HNS2}, the sporadic physical solutions really exist,
namely, $\lambda^{(2)}_0$ for $N=3$, and $\lambda^{(2)}_\pm$ for $N=6$.

(2) If $2s\equiv 0\pmod 2$, then if $\ell\equiv 1\pmod 2$, then total number of
states is a union of regular solutions and physical singular solutions, and if
$\ell\equiv 0\pmod 2$, then the total number of solutions is a union of
regular solutions and strange solutions.

(3) Let $\mathcal{N}_{strange}(N,\ell)$ (resp. $\mathcal{N}_{sp}(N,\ell)$)
be the total number of strange (resp. physical singular) solutions
corresponding to $N$ and $\ell$.
Then, if $2s\equiv 0\pmod 2$, we have \begin{align}
\mathcal{N}_{strange}(2N,2\ell)=\mathcal{N}_{sp}(2N-1,2\ell-1).
\end{align}
\qed
\end{conjecture}
An explicit, but still conjectural formula for the number
$\mathcal{N}_{sp}(2N,2\ell-1)$ has been stated in \cite{KS},
Conjecture 14 ({\bf B-b}), and we expect that the same conjecture
is valid for the numbers $\mathcal{N}_{sp}(2N-1,2\ell-1)$.

We would like to point out that this conjecture explains another difference
between integer spin chains and half-integer spin chains which attracts
great attention in the Haldane gap theory \cite{H}.

\paragraph{(d)}
Recently Deguchi--Giri \cite{DG} reported several results concerning the relationship
between the solutions to the Bethe ansatz equations and the rigged configurations.
Here we summarize some of their results which have close relation with our work.
(i) They compare our method with the so-called Bethe--Takahashi quantum number \cite{Bethe,T}
and confirmed the agreement.
(ii) They used the energy eigenvalues for the singular strings to distinguish several seemingly
similar solutions which correspond to different rigged configurations (see Figs. 1 and 2 of \cite{DG}).
(iii) They estimated the number of physical singular solutions from a different point of view.
They confirmed our previous conjecture in \cite{KS} for the even $\ell$ case.
However their result (Table VII of \cite{DG}) has discrepancy with the results of
\cite{HNS1,KS} which needs further study.


\begin{thebibliography}{99}
\bibitem[AV]{AV}
L.~V.~Avdeev and A.~A.~Vladimirov,
{\it Exceptional solutions to the Bethe ansatz equations},
Theor. Math. Phys. {\bf 69} (1986) 1071--1079.

\bibitem[BMSZ]{BMSZ}
N.~Beisert, J.~A.~Minahan, M.~Staudacher and K.~Zarembo,
{\it Stringing spins and spinning strings},
JHEP {\bf 09} (2003) 010 (27pp)

\bibitem[B]{Bethe}
H.~Bethe,
{\it Zur theorie der metalle},
Zeitschrift f\"{u}r Physik {\bf 71} (1931) 205--226.

\bibitem[DG]{DG} T.~Deguchi and P.~R.~Giri, 
\textit{Non self-conjugate strings, singular strings and rigged configurations in the Heisenberg model},
arXiv:1408.7030

\bibitem[EKS]{EKS}
F.~H.~L.~Essler, V.~E.~Korepin and K.~Schoutens,
{\it Fine structure of the Bethe ansatz for the
spin-$\frac{1}{2}$ Heisenberg $XXX$ model},
J. Phys. A: Math. Gen. {\bf 25} (1992) 4115--4126.

\bibitem[FT]{FT}
L.~D.~Faddeev and L.~A.~Takhtajan,
{\it Spectrum and scattering of excitations in
the one-dimensional isotropic Heisenberg model},
J. Soviet Math. {\bf 24} (1984) 241--267.

\bibitem[H]{H}
F.~D.~M.~Haldane,
{\it Nonlinear field theory of large-spin Heisenberg antiferromagnets:
Semiclassically quantized solitons of the one-dimensional easy-axis N\'{e}el state},
Phys. Rev. Lett., {\bf 50} (1983) 1153--1156.

\bibitem[HNS1]{HNS1}
W.~Hao, R.~I.~Nepomechie and A.~I.~Sommese,
{\it Completeness of solutions of Bethe's equations},
Phys. Rev. E {\bf 88} (2013) 052113 (8pp plus supplemental material).

\bibitem[HNS2]{HNS2}
W.~ Hao, R.~I.~Nepomechie and  A.~J.~Sommese, 
{\it Singular solutions, repeated roots and completeness for higher-spin 
chains},  J. Stat. Mech. (2014) P03024 (20pp)

\bibitem[KS]{KS}
A.~N.~Kirillov and R.~Sakamoto,
{\it Singular solutions to the Bethe ansatz equations and rigged configurations},
J. Phys. A: Math. Theor. {\bf 47} (2014) 205207 (20pp) 

\bibitem[KBI]{KorepinBook}
V.~E.~Korepin, N.~M.~Bogoliubov and A.~G.~Izergin,
{\it Quantum inverse scattering method and correlation functions},
Cambridge University Press (1993).

\bibitem[NW]{NW}
R.~I.~Nepomechie and C.~Wang,
{\it Algebraic Bethe ansatz for singular solutions},
J. Phys. A: Math. Theor. {\bf 46} (2013) 325002 (8pp).

\bibitem[T]{T}
M.~Takahashi,
{\it One-dimensional Heisenberg model at finite temperature},
Prog. Theor. Phys. {\bf 46} (1971) 401--415.

\end{thebibliography}
\end{document}